\documentclass[12pt,reqno]{amsart}
\usepackage{amsmath}
\usepackage{latexsym}
\usepackage{amsfonts}
\usepackage{amssymb}
\usepackage{color}
\usepackage{bbm,dsfont}
\usepackage{graphicx}

\newtheorem{definition}{Definition}
\newtheorem{example}{Example}
\newtheorem{remark}{Remark}
\newtheorem{theorem}{Theorem}
\newtheorem{lemma}{Lemma}
\newtheorem{corollary}{Corollary}

\newcommand{\hB}{\mathcal{B}}

\newcommand{\hF}{\mathcal{F}}

\newcommand{\hO}{\mathcal{O}} 

\newcommand{\hU}{\mathcal{U}}

\newcommand{\hW}{\mathcal{W}}

\newcommand{\N}{\mathbb N} 
\newcommand{\R}{\mathbb R} 
\newcommand{\Z}{\mathbb Z} 
\newcommand{\C}{\mathbb C} 
\newcommand{\T}{\mathbb T} 
\newcommand{\mo}[1]{\left| #1 \right|} 

\newcommand{\rmi}{\mathrm{i}}
\newcommand{\rme}{\mathrm{e}}

\newcommand{\fii}{\varphi} 
\newcommand{\Om}{\Omega} 
\newcommand{\la}{\lambda} 

\newcommand{\hi}{\mathcal{H}} 
\newcommand{\ki}{\mathcal{K}} 
\newcommand{\li}{\mathcal{L}} 
\newcommand{\mi}{\mathcal{M}} 

\newcommand{\lh}{\mathcal{L(H)}} 
\renewcommand{\th}{\mathcal{T(H)}} 
\newcommand{\sh}{\mathcal{S(H)}} 
\newcommand{\uh}{\mathcal{U(H)}} 

\newcommand{\tr}[1]{\mathrm{tr}\left[#1\right]} 
\def\<{\langle} 
\def\>{\rangle} 
\newcommand{\kb}[2]{|#1 \rangle\langle #2|} 
\newcommand{\ip}[2]{\left\langle #1 | #2 \right\rangle} 

\newcommand{\set}[2]{\big\{#1\,\big|\,#2\big\}} 
\newcommand{\bo}[1]{\mathcal{B}\left(#1\right)} 
\newcommand{\br}{\mathcal B(\mathbb R)} 
\newcommand{\lin}{\mathrm{lin}\,} 
\def\d{{\mathrm d}} 
\newcommand{\ov}{\overline} 
\newcommand{\CHI}[1]{\ensuremath{ \chi\raisebox{-1ex}{$\scriptstyle #1$} }} 
\newcommand{\hd}{\int_\Omega^\oplus\mi(x)\d\mu(x)} 
\newcommand{\rank}{\mathrm{rank}\,} 
\newcommand{\G}{{\hat G}} 
\renewcommand{\H}{{\mathcal{H}(\lambda)}} 
\renewcommand{\bo}[1]{\Sigma_{#1}} 
\newcommand{\HS}{\ki_{\rm HS}^\oplus} 

\newcommand{\Qo}{\mathsf{Q}} 
\newcommand{\Po}{\mathsf{P}} 
\newcommand{\Mo}{\mathsf{M}} 

\begin{document}

\title[Generalized Coherent States and POVMs]{Generalized Coherent States and Extremal Positive Operator Valued Measures}

\author[Heinosaari \and Pellonp\"a\"a]{Teiko Heinosaari \and Juha-Pekka Pellonp\"a\"a \\ \\ \tiny{Turku Centre for Quantum Physics \\ Department of Physics and Astronomy \\ University of Turku \\ Finland}}

\begin{abstract}
We present a correspondence between positive operator valued measures (POVMs) and sets of generalized coherent states.
Positive operator valued measures describe quantum observables and, similarly to quantum states, also quantum observables can be mixed.
We show how the formalism of generalized coherent states leads to a useful characterization of extremal POVMs.
We prove that covariant phase space observables related to squeezed states are extremal, while the ones  related to number states are not extremal.
\end{abstract}

\maketitle


\section{Introduction}

Coherent states have become a standard tool in quantum optics and many other areas \cite{CS85}.
In this paper we are interested in their relation to positive operator valued measures (POVMs).
Every student of quantum mechanics learns some examples of positive operator valued measures related to coherent states, perhaps even without hearing this term.
The most prominent example is the Q-function, first introduced by Husimi \cite{Husimi40}.
The coherent states in this case are the minimal uncertainty states of position and momentum, and the related POVM is defined by an integral over the corresponding rank-1 operators \cite{OQP97}.

Another common example is the spin direction observable \cite{BuSc89}. 
Suppose that we choose a direction randomly and measure a spin component in that direction.
We repeat this procedure, and as a result we have implemented a measurement of the spin direction.
The corresponding POVM is related to SU(2)-coherent states  \cite{Appleby00}.

In these two examples the coherent state structure of the POVMs involved is transparent. 
One could think that these are very special examples and that POVMs can have  quite diverse structure.
However, the canonical coherent states can be generalized in various ways \cite{CSWTG00}, and it turns out that any POVM arises from a set of generalized coherent states.

The set of all POVMs with a fixed outcome set and taking values in a fixed Hilbert space is convex.
A mixture of two POVMs corresponds to a procedure where two measurements are randomly alternated.
An extremal POVM is free from this type of randomization.
Their importance has been emphasized e.g. in \cite{SSQT01}.

In this paper we present a correspondence between POVMs and sets of generalized coherent states.
We show how this structure leads to a useful characterization of extremal POVMs.
Various aspects of extermal POVMs have been studied e.g. in \cite{Parthasarathy99,DaLoPe05,ChDaSc10,Pellonpaa11}.

We start by fixing the notation and recalling some concepts in Section \ref{sec:prel}.
In Section \ref{sec:cpso} we show that the POVM corresponding to the Husimi Q-function is extremal but that there are also other extremal POVMs related to the Weyl-Heisenberg group.
In Section \ref{sec:povm} we explain how every POVM arises from a set of generalized coherent states, and in Section \ref{sec:covariant} we show the consequences of a covariance property under symmetry transformations.

\section{Preliminaries}\label{sec:prel}

In the following we will need several different vector spaces and they are all assumed to be complex vector spaces.
Throughout this article, $\hi$ and $\mi$ are {separable} Hilbert spaces.
For any Hilbert space
$\hi$ we let $\lh$, $\uh$, $\th$, $\sh$ denote the set of bounded, unitary, trace class, and density operators on $\hi$, respectively. Elements of $\sh$ are identified with the {\it states} of a quantum system with the Hilbert space $\hi$.
The identity operator of a Hilbert space $\hi$ is denoted by $I_\hi$.

Let $(\Omega,\Sigma)$ be a measurable space (i.e.\ $\Sigma$ is a $\sigma$-algebra of subsets of a set $\Omega$). As usual, we define an empty sum to be 0, e.g.\ $\sum_{k=1}^0(\ldots):=0$, $\N:=\{0,1,\ldots\}$, and $\N_\infty:=\N\cup\{\infty\}$.
A mapping $\Mo:\,\Sigma\to\lh$ is called \emph{operator valued measure} if it is (ultra)weakly $\sigma$-additive.
We say that $\Mo$ is
\begin{itemize}
\item \emph{positive} if $\Mo(X)\geq 0$ for all $X\in\Sigma$,  
\item \emph{normalized} if $\Mo(\Omega)=I_\hi$
\item \emph{projection valued} if $\Mo(X)^2=\Mo(X)^*=\Mo(X)$ for all $X\in\Sigma$.
\end{itemize}
Commonly, a normalized projection valued measure is called a \emph{spectral measure}.

Normalized positive operator valued measures (POVMs) are identified with quantum observables, whereas normalized projection valued measures (PVMs) are called \emph{sharp observables} \cite{OQP97}.
For any POVM $\Mo:\,\Sigma\to\lh$ and $\rho\in\th$, we define a complex measure $\rho^\Mo:\,\Sigma\to\C$ by 
\begin{equation}
\rho^\Mo(X):=\tr{\rho\Mo(X)}, \qquad X\in\Sigma \, .
\end{equation}
 If $\rho\in\sh$ then $\rho^\Mo$ is a probability measure and it describes the statistics of a measurement of $\Mo$ in the state $\rho$.

We denote by $\hO(\Sigma,\,\hi)$ the set of all POVMs $\Mo:\,\Sigma\to\lh$.
This set is convex; if $\Mo_1,\Mo_2\in\hO(\Sigma,\,\hi)$, then $t\Mo_1 + (1-t)\Mo_2\in\hO(\Sigma,\,\hi)$ for every $0 < t < 1$.
An element $\Mo\in\hO(\Sigma,\,\hi)$ is \emph{extremal} if it cannot be written as a nontrivial convex decomposition, i.e., if $\Mo=t\Mo_1 + (1-t)\Mo_2$ for $0<t<1$, then $\Mo_1=\Mo_2=\Mo$.
From the physical point of view extremal observables describe quantum measurements that are free from any classical randomness, just in the same way as
pure states describe preparation procedures without classical randomness.

\section{Covariant phase space observables}\label{sec:cpso}

In this section we study covariant phase observables on the phase space $\R\times\R$.
We identify $\R\times\R\simeq\C$ and thus $\Omega=\C$ and $\Sigma$ is the Borel $\sigma$-algebra of $\C$.
The phase space shifts are given by the unitary operators $D(z)$, defined as
\begin{equation}
D(z)=e^{za^\dagger - \bar{z} a} \, ,
\end{equation}
where $a^\dagger,a$ are the usual creation and annihilation operators on $\hi=L^2(\R)$.

Let $h_n\in L^2(\R)$ be the $n$th normalized Hermite function for every $n=0,1,2,\ldots$.
They are eigenvectors of the number operator $N:=a^\dagger a$ and usually called \emph{number states}.
Another important family of vector states are the \emph{(canonical) coherent states}, parametrized by the complex numbers $z\in\C$ and defined as
\begin{equation}
\eta_0(x) := h_0(x)=\frac1{\sqrt[4]\pi} e^{- x^2/2} \, , \qquad  \eta_z := D(z)\eta_0 \, .
\end{equation}
\emph{The vacuum state} $\eta_0=h_0$ is both a number state and a coherent state, but otherwise these families of vector states are disjoint.

Using the canonical coherent states we can define a function 
\begin{equation}
Q_\varphi(z) :=  \mo{\ip{\varphi}{\eta_z}}^2
\end{equation}
for every unit vector $\varphi\in\hi$, and this is known as the \emph{$Q$-function} of a vector state $\varphi$.
Alternatively, we can write a POVM 
\begin{equation}
\Mo_{0} (Z) := \int_Z \kb{\eta_z}{\eta_z} \ \frac{\d^2 z}{\pi} = \int_Z D(z) \kb{\eta_0}{\eta_0} D(z)^* \ \frac{\d^2 z}{\pi} \, , \qquad Z\in\Sigma \, .
\end{equation}
Then 
\begin{equation}
\ip{\varphi}{\Mo_0(Z)\varphi}=\int_Z Q_\varphi(z) \ \frac{\d^2 z}{\pi} 
\end{equation}
and hence $\Mo_0$ describes the measurement of the $Q$-function.

More generally, for each density operator $\rho$ we can define a POVM $\Mo_\rho$ by
\begin{equation}
\Mo_\rho(Z):=\int_Z D(z)\rho D(z)^*\ \frac{\d^2 z}{\pi}  \, , \qquad Z\in\Sigma \, .
\end{equation}
These are called \emph{covariant phase space observables} and their properties have been analyzed in \cite{Appleby99JMO,CaDeLaLe00JMP}.
A covariant phase space observable can be interpreted as joint measurement of unsharp position and momentum observables \cite{Busch85}, and each of them can be implemented as a sequential measurement of an unsharp position measurement followed by a momentum measurement \cite{CaHeTo11}.

Obviously, a covariant phase space observable $\Mo_\rho$ can be extremal only if $\rho$ is a one-dimensional projection since otherwise the spectral decomposition of $\rho$ gives a nontrivial convex decomposition of $\Mo_\rho$ into other covariant phase space observables.
If $\rho$ is a one-dimensional projection, then $\Mo_\rho$ is extremal in the set of all covariant phase space observables.
However, we are interested on the extremality in the set of all POVMs.

We start with some preliminary observations.
Let $\psi\in\hi$ be a unit vector and $\kb{\psi}{\psi}$ the corresponding one-dimensional projection.
We recall from \cite{Pellonpaa11} that a POVM $\Mo_{\kb{\psi}{\psi}}$ is extremal if and only if, for every $\lambda\in L^\infty(\C)$, we have
\begin{align}\label{eq:ext-lambda}
\int_\C \lambda(w) D(w)\kb{\psi}{\psi}D(w)^\ast \d^2 w = 0  \qquad \Longrightarrow \qquad \lambda = 0 \, .
\end{align}
For $A\in\lh$, we have $\ip{\eta_z}{A\eta_z}=0$ for every $z\in\C$ if and only if $A=0$ \cite{LaPe02}.
We thus conclude the following.

\begin{lemma}\label{lemma:lambda}
Let $\psi\in\hi$ be a unit vector.
A POVM $\Mo_{\kb{\psi}{\psi}}$ is extremal if and only if for every $\lambda\in L^\infty(\C)$, we have
\begin{align}\label{eq:ext-lambda}
\int_\C \lambda(w) \mo{\ip{\eta_z}{D(w)\psi}}^2 \d^2 w = 0 \quad \forall z\in\C \qquad \Longrightarrow \qquad \lambda = 0 \, .
\end{align}
\end{lemma}

We will now utilize this criterion to see that some covariant phase space observables are extremal and some are not.
For each $z\in\C$, we obtain
\begin{align*}
 \int_\C \lambda(w) \mo{\ip{\eta_z}{D(w)\psi}}^2 \frac{\d^2 w}{\pi} =  \int_\C \lambda(w) \mo{\ip{\eta_{z-w}}{\psi}}^2 \frac{\d^2 w}{\pi} = ( \lambda \ast Q_\psi ) (z) \, , 
\end{align*}
where $\lambda$ in the last expression is a tempered distribution defined by the function $\lambda\in L^\infty(\C)$ and $*$ is the convolution.
(Note that $Q_\psi$ is a smooth rapidly decreasing function $\R^2\to\R$, hence the convolution can be defined.)

The fact that $\Mo_0$ is extremal has been pointed out already in \cite{Holevo85}.
We can now present a full characterization of extremal covariant phase space observables.

\begin{theorem}\label{thm:extremal}
Let $\psi\in\hi$ be a unit vector.
The following conditions are equivalent.
\begin{itemize}
\item[(i)] The covariant phase space observable $\Mo_{\kb{\psi}{\psi}}$ is extremal in the set of all POVMs.
\item[(ii)] The Fourier transform of $Q_\psi$ is everywhere nonzero.
\item[(iii)] $\ip{\psi}{D(z)\psi}\neq 0$ for every $z\in\C$.
\end{itemize}
\end{theorem}

\begin{proof}
(i)$\Leftrightarrow$(ii):
The condition $( \lambda \ast Q_\psi ) (z)=0$ for all $z\in\C$ equals $ \widehat{Q_\psi}\widehat{\lambda}=0$
 where $\widehat{}$ means the Fourier transform of a tempered distribution  \cite{FA91}.
By Lemma \ref{lemma:lambda} we conclude that $\Mo_{\kb{\psi}{\psi}}$ is extremal if and only if $ \widehat{Q_\psi}\widehat{\lambda}=0$, $\lambda\in L^\infty(\C)$, implies $\lambda=0$.
If $ \widehat{Q_\psi}$ is everywhere nonzero then $ \widehat{Q_\psi}\widehat{\lambda}=0$ implies that $\widehat{\lambda}=0$ and hence $\lambda=0$.
Suppose then that $ \widehat{Q_\psi}(w)= 0$ for some $w\in\C$. By choosing $\lambda(z)=\exp[\rmi(z\ov w+\ov z w)]$ one sees that $\widehat\lambda$ proportional to a Dirac delta distribution concentrated on $w$ and 
$ \widehat{Q_\psi}\widehat{\lambda}=0$ yielding a contradiction.
(ii)$\Leftrightarrow$(iii): 
The Q-function is a Gaussian convolution of the Wigner function.
The Wigner function of a vector state $\psi$ is the Fourier transform of the function $z\mapsto\ip{\psi}{D(iz)\psi}$.
Hence, the Fourier transform $Q_\psi$ is everywhere nonzero if and only if the function $z\mapsto\ip{\psi}{D(z)\psi}$ is everywhere nonzero.
\end{proof}

As an easy consequence of Theorem \ref{thm:extremal} we see that all covariant phase space observables related to coherent states are extremal.
Namely, if $\psi=\eta_w$ for some $w\in\C$, then
\begin{align}
Q_{\eta_w}(z)=\mo{ \ip{\eta_z}{\eta_w} }^2 = e^{-\mo{w-z}^2} \, .
\end{align}
Since $Q_{\eta_w}$ is a Gaussian function, its Fourier transform is also Gaussian and therefore everywhere nonzero.
More generally, let $\psi$ be a \emph{squeezed state}, in which case $\psi$ is a Gaussian function.
Then again $Q_{\psi}$ is a Gaussian function, its Fourier transform is Gaussian and thus everywhere nonzero.
We conclude that $\Mo_{\kb{\psi}{\psi}}$ is extremal if $\psi$ is a squeezed state.

Let us then have examples of covariant phase space observables that are not extremal even if they are extremal in the set of covariant phase space observables.
The Fourier transform of $f\in L^1(\C)$ can be written in the form
$$
\widehat f(w)=\frac{1}{\pi}\int_\C f(z)\rme^{-\rmi(z\ov w+\ov z w)}\d^2 z=\frac{1}{2\pi}\int_\R\int_\R f\big(2^{-1/2}(q+\rmi p)\big)
\rme^{-\rmi(uq+vp)}\d q\d p
$$
where $w=(u+\rmi v)/\sqrt2$ and $z=(q+\rmi p)/\sqrt2$.
Thus, for
\begin{align*}
Q_{h_n}(z)=  \frac{e^{-\mo{z}^2} \mo{z}^{2n}}{n!} \, 
\end{align*}
we obtain
$$
\widehat{Q_{h_n}}(w)=L_{n}\big(|w|^2\big)\rme^{-|w|^2}
$$
where
$$
L_n(t):=\frac{\rme^{t}}{n!}\frac{\d^n t^n\rme^{-t}}{\d t^n} 
$$ 
is the $n$th Laguerre polynomial which has $n$ (real) strictly positive roots.
For example,
$$
\widehat{Q_{h_1}}(w)= L_{1}\big(|w|^2\big)\rme^{-|w|^2}=\big(1-|w|^2\big)\rme^{-|w|^2}
$$
so that $\widehat{Q_{h_1}}(w)=0$ when $\mo{w}=1$.

We can also find a convex decomposition of $\Mo_{\kb{h_1}{h_1}}$.
For $\lambda(z)=\cos(z+\ov z)$ one gets
\begin{align*}
\int_\C \lambda(w) \mo{\ip{\eta_z}{D(w)h_1}}^2 \d^2 w
& =
\int_\C \cos(w+\ov w) \rme^{-\mo{z-w}^2} \mo{z-w}^{2} 
 \d^2 w = 0 
\end{align*}
so that we may write
$
\Mo_{\kb{h_1}{h_1}}=\frac{1}{2}\Mo_+ + \frac{1}{2}\Mo_-
$
where the (unequal) POVMs $\Mo_{\pm}$ are defined by
$$
\Mo_{\pm}(Z):=\frac{1}{\pi}\int_Z \big[1\pm\cos(z+\ov z)\big]\, D(z) \kb{h_1}{h_1} D(z)^*\d^2 z \, , \qquad Z\in\Sigma \, .
$$

Finally, we notice a connection of extremality to the informational completeness \cite{Prugovecki77}.
A POVM $\Mo$ is informationally complete if $\rho_1^\Mo\neq\rho_2^\Mo$ for two states $\rho_1\neq\rho_2$.
It was proved in \cite{AlPr77JMP} that a covariant phase space observable $\Mo_\rho$ is informationally complete if $\tr{\rho D(z)}\neq 0$ for almost all $z\in\C$.
Hence, from Theorem \ref{thm:extremal} we conclude the following.

\begin{corollary}\label{cor:ic}
Let $\psi\in\hi$ be a unit vector.
If the covariant phase space observable $\Mo_{\kb{\psi}{\psi}}$ is extremal in the set of all POVMs, then it is informationally complete.
\end{corollary}

Obviously, a POVM can be extremal without being informationally complete.
For instance, any projection valued measure is extremal but not informationally complete.

\section{Generalized coherent states for POVMs}\label{sec:povm}

As we have seen in Section \ref{sec:cpso}, the extremality question for covariant phase space observables  turns out to be tractable due to their simple mathematical structure.
Here we show that a similar approach can be developed also in the general situation.

Let ${\bf h}=\{h_n\}_{n=1}^{\dim\hi}$ 
be an orthonormal basis of $\hi$ 
and denote $V_{\bf h}:=\lin\{h_n\}.$
Then $V_{\bf h}$ is a dense subset in $\hi$.
Let $V_{\bf h}^\times$ be the algebraic antidual of the vector space $V_{\bf h}$. Recall that $V_{\bf h}^\times$ can be identified with the linear space of formal series $c=\sum_{n=1}^{\dim\hi}c_nh_n$ where $c_n$'s are arbitrary complex numbers. 
Hence, $V_{\bf h}\subseteq\hi\subseteq V_{\bf h}^\times$.

For all $\psi\in V_{\bf h}$ and $c\in V_{\bf h}^\times$, we denote their dual pairing by
\begin{equation}
\<\psi|c\>:=\sum_{n=1}^{\dim\hi} \<\psi|h_n\>c_n \quad \textrm{and} \quad \<c|\psi\>:=\overline{\<\psi|c\>} \, .
\end{equation}
We say that a mapping 
$$
c:\,\Omega\to V_{\bf h}^\times \, , \quad c(x)= \sum_{n=1}^{\dim\hi}c_n(x)h_n
$$
is {\it (weak${}^*$-)measurable} if all its components $x\mapsto c_n(x)$ are measurable. Note that, if $c:\,\Omega\to\hi\subseteq V_{\bf h}^\times$ is weak${}^*$-measurable then
the maps $x\mapsto\<\psi|c(x)\>$ are measurable for all $\psi\in\hi$.

We denote by $\hd$ a direct integral Hilbert space where 
$\mu$ is a $\sigma$-finite nonnegative measure on $(\Omega,\Sigma)$ and
$\{\mi(x)\}_{x\in\Om}$ is a ($\mu$-measurable) field of {separable} Hilbert spaces. Recall that for any $\psi\in\hd$ one has $\psi(x)\in\mi(x)$ for all $x\in\Om$ and $\|\psi\|^2=\int_\Omega\|\psi(x)\|^2_{\mi(x)}\d\mu(x)$.
In addition, the mapping $\Om\ni x\mapsto\dim\mi(x)\in\N_\infty$ is $\mu$-measurable.
For example, the Hilbert space $L^2_\mi(\mu)$ of square $\mu$-integrable functions from $\Omega$ to$\mi$ is of the direct integral form where now $\mi(x)\equiv\mi$.

We denote by $L^\infty(\mu)$ the Abelian von Neumann algebra of $\mu$-essentially bounded complex functions on $\Omega$, and for any  $f\in L^\infty(\mu)$, we denote by $\tilde f$ the multiplicative (i.e.\ diagonalizable) operator $(\tilde f\psi)(x):=f(x)\psi(x)$ on $\hd$. 
Especially, we have the {\it canonical spectral measure} 
$$
\Sigma\ni X\mapsto\widetilde{\CHI X}\in\li\Big(\hd\Big) \, , 
$$
where $\CHI X$ is the characteristic function of $X\in\Sigma$.
A bounded operator $D$ on $\hd$ is decomposable, i.e.\ $D=\int_\Omega^\oplus D(x)\d\mu(x)$, if $D(x)\in\mathcal L(\hi_{n(x)})$ and $(D\psi)(x)=D(x)\psi(x)$ for $\mu$-almost all $x\in\Omega$ and for all  $\psi\in\hd$, and
$\|D\|={\rm ess\,sup}_{x\in\Omega}\|D(x)\|<\infty$. Obviously, any $f\in L^\infty(\mu)$ defines a bounded decomposable operator $\tilde f$.

The results proved in \cite{Pellonpaa11,HyPeYl07} can be now stated in the following form.

\begin{theorem} \label{thrm:peruslause}
Let $\Mo:\,\Sigma\to\lh$ be a POVM, $\mu:\,\Sigma\to[0,\infty]$ a $\sigma$-finite measure such that $\Mo$ is absolutely continuous with respect to $\mu$, and ${\bf h}$ be an orthonormal basis of $\hi$.
There exists a direct integral Hilbert space $\mi^\oplus\equiv\hd$ (with $\dim\mi(x)\le\dim\hi$) 
such that for all $X\in\Sigma$,
\begin{enumerate}
\item
$\Mo(X)=Y^*\widetilde{\CHI X}Y$ where $Y=\sum_{m=1}^{\dim\hi}\kb{\psi_m}{h_m}$ is an isometry and
$\{\psi_m\}_{m=1}^{\dim\hi}$ is an orthonormal set of $\mi^\oplus$ such that
the set of linear combinations of vectors $\CHI X \psi_m$ is dense in $\mi^\oplus$
(a minimal Naimark dilation of $\Mo$). Hence,
\begin{eqnarray*}
\Mo(X)=\sum_{n,m=1}^{\dim\hi}\int_{X}\<\psi_n(x)|\psi_m(x)\>_{\mi(x)}\d\mu(x)\ \kb{h_n}{h_m} \qquad \text{(weakly).}
\end{eqnarray*} 
\item
There are measurable maps $d_k:\,\Omega\to V_{\bf h}^\times$ such that, for all $x\in\Omega$, 
the vectors $d_k(x)\ne 0$, $k\in\Z_{\dim\mi(x)}$, are linearly independent, and
$$
\<\fii|\Mo(X)\psi\>=\int_X \sum_{k=1}^{\dim\mi(x)} \<\fii|d_k(x)\>\<d_k(x)|\psi\>\d\mu(x),\hspace{0.5cm}\fii,\,\psi\in V_{\bf h},
$$
(a minimal diagonalization of $\Mo$). 
\item $\Mo$ is a spectral measure if and only if $\{\psi_m\}_{m=1}^{\dim\hi}$ is an orthonormal basis of $\mi^\oplus$. Then $Y$ is a unitary operator and $\mi^\oplus$ can be identified with $\hi$.
\item $\Mo$ is extremal if and only if for any decomposable operator $D=\int_\Omega^\oplus D(x)\d\mu(x)\in\li\big(\mi^\oplus\big)$, the condition $Y^*DY=0$ implies $D=0$.
\end{enumerate}
\end{theorem}

\begin{remark}\rm \label{rem:antidual}
As shown in \cite[Section 6]{HyPeYl07},
$V_{\bf h}$ can be extended to a Banach space $\hB$ 
such that the vectors $d_k(x)$ in (2) of theorem \ref{thrm:peruslause} can be viewed as elements of the topological antidual $\hB^\times$ of $\hB$. Hence, throughout this article, one can replace $V_{\bf h}$ and $V_{\bf h}^\times$ with $\hB$ and $\hB^\times$, respectively.
Especially, we have a triplet of Banach spaces $\hB\subseteq\hi\subseteq \hB^\times$.
\end{remark}

\begin{definition}\rm  (\cite{Pellonpaa11}) 
We say that $Y\hi\subseteq\mi^\oplus$ is an {\it $\Mo$-representation} space. It is unique up to a decomposable unitary map so that we may choose $$\mi(x)=\lin\set{h_n}{n \le \dim\mi(x)}.$$ The elements of $Y\hi$ are called {\it $\Mo$-wave functions}.
The vectors $d_k(x)$ in (2) of theorem \ref{thrm:peruslause} are called {\em generalized coherent states} associated with the POVM $\Mo$. If, for $\mu$-almost all $x\in\Omega$, $n(x)\in\{0,r\}$ where $r\in\N_\infty$ we say that $\Mo$ {\em is of a constant rank,} and denote $\rank\Mo:=r$.
\end{definition}

If a POVM $\Mo$ is of a constant rank then, by defining $$\Omega_{\rm supp}:=\set{x\in\Omega}{\dim\mi(x)>0},$$
we can write
$$
\<\fii|\Mo(X)\psi\>=\sum_{k=1}^{\rank\Mo}\int_{X\cap\Omega_{\rm supp}} \<\fii|d_k(x)\>\<d_k(x)|\psi\>\d\mu(x),\hspace{0.5cm}\fii,\,\psi\in V_{\bf h}.
$$
Notice that one can always define $d_k(x):=0$ for all $k>\dim\mi(x)$ so that one may let the $k$--index in (2) of theorem \ref{thrm:peruslause} run from 1 to $\infty$ and then take the summation out from the integrand as above. Finally, we note that the relation between vectors $d_k(x)\in V_{\bf h}^\times$ and $\psi_n(x)\in\mi(x)$ can be chosen to be 
$$
\<d_k(x)|h_n\>=\<b_k(x)|\psi_n(x)\>_{\mi(x)}=\<b_k(x)|(Yh_n)(x)\>_{\mi(x)}
$$
where $x\mapsto\{b_k(x)\}_{k=1}^{\dim\mi(x)}$ is a measurable field of orthonormal bases of $\hd$. Hence, by remark \ref{rem:antidual}, when $V_{\bf h}^\times$ is replaced by $\hB^\times$ we have
\begin{equation}\label{eq:tarkea}
\<d_k(x)|\psi\>=\<b_k(x)|(Y\psi)(x)\>_{\mi(x)}
\end{equation}
for $\mu$-almost all $x\in\Omega$ and for all $\psi\in\hB$.

\begin{remark}\rm
Let $S$ be a (bounded or unbounded) selfadjoint operator on $\hi$ and $\mathsf M$ its spectral measure (defined on the Borel $\sigma$-algebra $\br$ of $\R$). 
As shown in \cite{HyPeYl07}, there exists an orthonormal basis ${\bf h}$ of $\hi$ such that $SV_{\bf h}\subseteq V_{\bf h}$ and, if $S^\times:\,V_{\bf h}^\times\to V_{\bf h}^\times$ is an extension of $S$,  one gets $S^\times d_k(x)=xd_k(x)$ for $\mu$--almost all $x$ in the spectrum of $S$. Hence, theorem \ref{thrm:peruslause} can be viewed as a generalization of {\it Dirac formalism} for POVMs and we may call $\dim\mi(x)$ the multiplicity of a measurement outcome $x\in\Omega$.
\end{remark}

\section{Generalized coherent states for covariant POVMs}\label{sec:covariant}

In the following we study a situation where a symmetry group $G$ is a separable locally
compact unimodular group of type I. 
Examples of such groups are compact groups (e.g.\ SU($n$)), locally compact Abelian groups (e.g.\ $\R^n$),
connected semi-simple Lie groups and nilpotent Lie groups (e.g.\ Heisenberg group) so that
the usual symmetry groups used in physics fall into this class.

The behavior of a POVM under the symmetry group determines to a large extent its structure and
this idea goes back to Mackey \cite{MFQM63}.
For simplicity, we assume that $G$ acts on itself from left, hence the value space of POVMs is $G$ and we denote by $\bo G$  the Borel $\sigma$-algebra of $G$. 
Our results on covariance systems is based on theory developed in \cite{Cattaneo79,Holevo87,CaDe03,HoPe09}.

Let $V:\,G\to\uh,\,g\mapsto V_g$, be a strongly continuous unitary representation of $G$ on a separable
Hilbert space $\hi$. Recall that $V$ is always decomposable, that is, $\hi$ can be
represented as a direct integral of Hilbert spaces and the
representation `maps a fiber into the same fiber'. 
More specifically, we assume that 
\begin{equation}
\mathcal{H}=\int_{\G}^{\oplus }\mathcal{H}{(\la )}\d\nu (\la ),\qquad
\mathcal{H }{(\la )=\mathcal{K}(\la )}\otimes \mathcal{L}{(\la )},
\end{equation}
where $\G$ is the unitary dual of $G$, that is, for any $\la \in \hat{G}$, the map $g\mapsto \la
_{g}$ is an irreducible representation acting in a Hilbert space
${\mathcal{K}(\la )}$
and $\mathcal{L}(\la )$ is the multiplicity (Hilbert) space of $\la$. Moreover,
\begin{equation}
V_{g}=\int_{\G}^{\oplus }V_{g}(\la )\d\nu(\la ),\qquad 
V_{g}(\la):=\la _{g}\otimes I_{\mathcal{L}{(\la )}}.
\end{equation}
It is necessary for the existence of a covariant POVM
that the measure $\nu$ is absolutely continuous with respect to the
Plancherel measure and hence we may and will assume that 
{\it $\nu$ is the Plancherel measure on $\hat{G}$ associated with the Haar measure of $G$}.
We say that a POVM
$\Mo:\,\bo G\to\lh$ is \textit{covariant} if
$$
V_g \Mo(X)V_g^*=\Mo(gX), \quad g\in G, \quad X\in\bo G \, .
$$

For any separable Hilbert space $\mi$ we define $L^2_\mi(G):=L^2_\mi(\mu)$ where $\mu:\,\bo G\to[0,\infty]$ is a Haar measure of $G$ denoted briefly by $\d\mu(g)=\d g$.
Let $T:\,G\to\hU\big(L^2_\mi(G)\big),\,g\mapsto T_g$, be the strongly continuous unitary representation of $G$ defined by
\begin{equation}\label{eq:tildeV}
\big(T_g\psi\big)(g')=\psi\big(g^{-1}g'\big), \quad g,\,g'\in G,\quad \psi\in
L^2_{\mathcal{M}}(G).
\end{equation}
We have $\tilde V_g\widetilde{\CHI X}\tilde V_g^*=\widetilde{\CHI{gX}}$.
Now (1) of theorem \ref{thrm:peruslause} can be replaced by the following covariant version of a minimal Naimark dilation \cite[Proposition 2]{Cattaneo79}.

\begin{theorem}\label{thrm:covaNaima}
For any covariant POVM $\Mo$, there exist a separable Hilbert space $\mi$ (with $\dim\mi\le\dim\hi$) such that, for all $X\in\bo G$,
$\Mo(X)=Y^*\widetilde{\CHI X}Y$ where $Y$ is an isometry from $\hi$ into $L^2_\mi(G)$ for which 
$$
YV_g=T_gY,\qquad g\in G,
$$
and the set $\set{\CHI X Y \psi}{X\in\bo G,\,\psi\in\hi}$ is total in $L^2_\mi(G)$.
\end{theorem}
Immediately we see that now $\hd$ of Theorem \ref{thrm:peruslause} is $L^2_\mi(G)$ and thus any covariant POVM is of constant rank. More specifically, 
$\rank\Mo=\dim\mi.$
Next we determine the generalized coherent states $d_k(g)$ of theorem \ref{thrm:peruslause} for a covariant POVM $\Mo$ with the isometry $Y$ of Theorem \ref{thrm:covaNaima}.

By using the unitary Fourier-Plancherel transform
$\mathcal{F}$ of the representation $T$ and defining $W:=\mathcal{F}  Y$, $U_g:=\mathcal{F}  T_g \mathcal{F}^*$, and  $\Po(X):=\mathcal{F}\widetilde{\CHI X}\mathcal{F}^*$, we have
\begin{itemize}
\item $\Mo(X)=W^* \Po(X) W,\hspace{1cm}X\in\bo G,$
\item $W V_g= U_g W,\hspace{1cm}g\in G.$
\end{itemize}
Since $L^2_\mi(G)\cong L^2(G)\otimes\mi$ we may write the Fourier-Plancherel transform $\hF$ of the form $\hF=\hF_0\otimes I_\mi$ where
$\hF_0$ is the Fourier-Plancherel transform from $L^2(G)$ onto
the direct integral
$$
\HS=\int_\G^\oplus[\mathcal{K}(\la)\otimes\mathcal{K}(\la)^*]\d\nu(\la),
$$ where $\mathcal{K}(\la)^*$ is the dual of $\mathcal{K}(\la)$ and, hence,
$\mathcal{K}(\la)\otimes\mathcal{K}(\la)^*$ is isomorphic to the Hilbert
space of Hilbert-Schmidt operators\footnote{Via $\varphi\otimes\eta^*\mapsto\kb\eta\varphi$ where $\varphi,\,\eta\in\ki(\la)$ and $\eta^*=\<\eta|\in\ki(\la)^*$.}
 on $\mathcal{K}(\la)$ and
$$
\tilde\hi:=\HS\otimes\mi=\hF L^2_\mi(G)=\int_\G^\oplus\underbrace{[\mathcal{K}(\la)\otimes\mathcal{K}(\la)^*\otimes\mathcal{M}]}_{=:\;\tilde\hi(\la)}\d\nu(\la).
$$
Recall that, for each integrable $\psi\in L^2(G)$ one has 
$(\hF_0\psi)(\la)=\int_G \la_g\psi(g)\d g$ for (almost) all $\lambda\in\G$. In addition,
$(\hF_0^*\fii)(g)=\int_\G\tr{\fii(\la)\la_g^*}\d\nu(\la)$ where $\fii\in\HS$ is such that $\int_\G\|\fii(\la)\|_{\rm tr}\d\nu(\la)<\infty$
and $\|\,\cdot\,\|_{\rm tr}$ is the trace norm. Let $\{b_k\}_{k=1}^{\dim\mi}$ be an orthonormal basis of $\mi$.
Since any $\tilde\fii\in\hi$ can be written in the form $\tilde\fii=\sum_{k=1}^{\dim\mi}\fii_k\otimes b_k$, $\fii_k\in\HS$,
we may define 
$$
\|\tilde\fii\|_{\tilde\hB}:=\|\tilde\fii\|_{\tilde\hi}+
\int_\G \sqrt{\sum_k\|\fii_k(\la)\|_{\rm tr}^2}\d\nu(\la)
$$ 
and a Banach space
$
\tilde\hB:=\set{\tilde\fii\in\tilde\hi}{\|\tilde\fii\|_{\tilde\hB}<\infty}
$
which is dense in $\tilde\hi$.

It is easy to verify that, for
$$
U_g:=
\int_\G^\oplus[\la_g\otimes I_{\mathcal{K}(\la)^*}\otimes I_{\mathcal{M}}]\d\nu(\la)\in\hU(\tilde\hi),
$$ 
the condition $U_g\hF=\hF T_g$ holds and the restriction $U_g\big|_{\tilde\hB}:\,\tilde\hB\to\tilde\hB$ is continuous. 
Now the isometry $W:\,\hi\to\tilde\hi$ is decomposable, that is,
\begin{equation}
W=\int_\G^\oplus\big[I_{\mathcal{K}(\la)}\otimes W(\la)\big]\d\nu(\la),
\end{equation}
where $W(\la):\,\mathcal{L}(\la)\to \mathcal{K}(\la)^*\otimes\mathcal{M}$
is an isometry for $\nu$-almost all $\la\in\G$ \cite[Lemma 1]{HoPe09} and, for all $g\in G$, 
\begin{align*}
W V_g
&=\int_\G^\oplus\big[I_{\mathcal{K}(\la)}\otimes W(\la)\big]\big[\la _{g}\otimes I_{\mathcal{L}{(\la )}}\big]\d\nu(\la ) \\
&=\int_\G^\oplus\big[\la_g\otimes I_{\mathcal{K}(\la)^*}\otimes I_{\mathcal{M}}\big]\big[I_{\mathcal{K}(\la)}\otimes W(\la)\big]\d\nu(\la)
=U_g W.
\end{align*}

In the spirit of remark \ref{rem:antidual}, we define a Banach space
$$
\hB:=\Big\{\psi\in\hi\,\Big|\,W\psi\in\tilde\hB\Big\}
$$ 
equipped with the norm $\|\psi\|_\hB:=\|W\psi\|_{\tilde \hB}$, $\psi\in\hB$, so that $W\big|_\hB:\,\hB\to\tilde\hB$ is trivially continuous.
Now $\hB$ is dense in $\hi$ (and thus contains an orthonormal basis $\bf h$ of $\hi$) and, for each $g\in G$, $\|V_g\psi\|_\hB=\|\psi\|_\hB$ for all $\psi\in\hB$, so that the restriction $V_g\big|_\hB:\,\hB\to\hB$ is continuous.

Let $k\in\Z_{\dim\mi}$, $g\in G$, and $\psi\in\hB$. 
Since $YV_g=T_gY$ and $Y=\hF^*W$, according to \eqref{eq:tarkea}, we define
$$
\<d_k(g)|\psi\>:=\<b_k|(\hF^*WV_g^*\psi)(e)\>_{\mi}
$$
where $e$ is the unit of $G$.
Indeed, since $V_g\big|_\hB:\,\hB\to\hB$ and $W\big|_\hB:\,\hB\to\tilde\hB$ are continuous
one needs only study innerproducts
$
\<b_k|(\hF^*\fii)(e)\>_{\mi}
$
where $\fii\in\tilde\hB$:
Let $\fii=\sum_{k=1}^{\dim\mi}\fii_k\otimes b_k\in\tilde\hB$ where $\fii_k\in\HS$ are thus such that
\begin{eqnarray*}
\|\fii\|_{\tilde\hB}=
\sqrt{\sum_{k}\int_\G\tr{\fii_k(\la)^*\fii_k(\la)}\d\nu(\la)}+\int_\G \sqrt{\sum_{k}\|\fii_k(\la)\|_{\rm tr}^2}\d\nu(\la)
<\infty.
\end{eqnarray*}
Now $\<b_k|(\hF^*\fii)(e)\>_{\mi}=(\hF_0^*\fii_k)(e)=\int_\G\tr{\fii_k(\la)}\d\nu(\la)$
so that
$$
\big|\<b_k|(\hF^*\fii)(e)\>_{\mi}\big|\le\int_\G|\tr{\fii_k(\la)}|\d\nu(\la)\le
\int_\G\|\fii_k(\la)\|_{\rm tr}\d\nu(\la)\le\|\fii\|_{\tilde\hB}
$$
and, hence, the linear mappings 
$
\hB\ni\psi\mapsto\<d_k(g)|\psi\>\in\C
$
are continuous and their conjugate (antilinear) maps
$
d_k(g):\,\hB\to\C,\;\psi\mapsto
\<\psi|d_k(g)\>:=\overline{\<d_k(g)|\psi\>}
$
belongs to $\hB^\times$. To conclude, we have
\begin{theorem}
For any covariant POVM $\Mo$ there exist a dense Banach space $\hB\subseteq\hi$ for which $V_g\big|_\hB:\,\hB\to\hB$ is continuous for all $g\in G$ and linearly independent vectors  $d_k\in\hB^\times$, $k\in\Z_{\rank\Mo}$, such that
$$
\<\fii|\Mo(X)\psi\>=\sum_{k=1}^{\rank\Mo} \int_X  \<V_g^*\fii|d_k\>\<d_k|V_g^*\psi\>\d g,\hspace{0.5cm}\fii,\,\psi\in\hB,\;X\in\bo G
$$
(a minimal covariant diagonalization of $\Mo$). 
\end{theorem}
From Theorem \ref{thrm:peruslause} we see that, if $\rank\Mo<\infty$ one gets a simple characterization of extremality in terms of coherent states: $\Mo$ is extremal if and only if, for any (essentially) bounded $\mu$-measurable functions $\lambda_{kl}:\,\Omega\to\C$,
$$
\sum_{k,l=1}^{\rank\Mo} \int_\Omega\lambda_{kl}(g) \<V_g^*\psi|d_k\>\<d_l|V_g^*\psi\>\d g=0,\hspace{0.5cm}\psi\in\hB,
$$
implies $\lambda_{kl}=0$ for all $k,\,l$. Next we consider an easy special case.

\begin{example}[\bf Abelian group]\rm
Assume that $G$ is Abelian group. Now $\G$ is the dual group consisting of characters $\lambda:\,G\to\T$ with the (dual) Haar measure
$\d\lambda$.
Hence $\d\nu(\lambda)=\d\lambda$ and, since any irreducible representation is one dimensional, $\ki(\la)\cong\C$, one has $\hi(\la)\cong\li(\la)$ in
$\hi=\int_\G^\oplus\H\d\lambda$. Moreover, 
$
g\mapsto V_g=\int_\G^\oplus V_g(\lambda)\d\lambda
$
where $V_g(\lambda)\psi(\lambda)=\lambda(g)\psi(\lambda)$ for $\d\lambda$-almost all $\lambda\in\G$.
By using the Fourier-Plancherel transform $\mathcal{F}:\,L_\mi^2(G)\to L_\mi^2(\G)$
(i.e.\ $(\mathcal{F}\fii)(\lambda)=\int_G\lambda(g)\fii(g)\d g$) and a spectral measure $X\mapsto\Po(X)=\mathcal{F}  \widetilde{\CHI X}  \mathcal{F}^*$,
one gets for any covariant POVM $\Mo:\,\bo G\to\lh$ that $\Mo(X)=W^* \Po(X) W$ where the isometry
$
W:\,\hi\to\tilde\hi\cong L^2_{\mathcal{M}}(\G)
$
is decomposable (i.e.\
$
W=\int_{\Lambda}^\oplus W(\lambda)\d\lambda
$
where the operators $W(\lambda):\,\H\to\mathcal{M}$ are isometries). 
In addition,
$$
\|\psi\|_{\hB}=\sqrt{\int_\G\|\psi(\la)\|_{\hi(\la)}^2\d\la}+\int_\G\|\psi(\la)\|_\H\d\la
$$
and $\hB=\set{\psi\in\hi}{\|\psi\|_{\hB}<\infty}$.
Hence, for all $\fii,\,\psi\in\hB$ and $X\in\bo G$,
\begin{align*}
\<\fii|\Mo(X)\psi\>&=\<W\fii|\Po(X)W\psi\> \\
&= \int_X\int_\G\int_\G \lambda(g)\ov{\lambda'(g)}
\<W(\lambda)\fii(\lambda)|W(\lambda')\psi(\lambda')\>_{\mi}\d\lambda\d\lambda' \d g \\
&=
\int_X\<\hW V_g^*\fii|\hW V_g^*\psi\>_{\mi} \d g
\end{align*}
where 
$
{\mathcal W}:\,\hB\to{\mathcal M},\;\psi\mapsto{\mathcal W}\psi:=(\hF^*W\psi)(e)=\int_\G W(\lambda)\psi(\lambda)\d\lambda
$
is linear and continuous (since $\|\hW\psi\|_{\mathcal M}\le\|\psi\|_\hB$).
Since $(Y\psi)(g)=\hW V_g^*\psi$ for all $\psi\in\hB$ and (almost) all $g\in\G$ and the set $\set{\CHI X Y \psi}{X\in\bo G,\,\psi\in\hi}$ is total in $L^2_\mi(G)$ one sees that
$\hW\hB$ is total in $\mi$. 

Let $\{b_k\}_{k=1}^{\rank\Mo}$ be an orthonormal basis of $\mi$.
Using the fact that the restrictions $V_g\big|_{\hB}$ are continuous on $\hB$ one can define, for all $g\in G$ and $k\le\rank\Mo$, a continuous linear map $\<d_k(g)|:\,\hB\to\C$ by
$$
\<d_k(g)|\psi\>:=\<b_k|(\hF^*WV_g^*\psi)(e)\>_\mi=\<b_k|\hW V_g^*\psi\>_{\mi},\qquad\psi\in\hB,
$$
so that $d_k(g)\in\hB^\times$ where $[d_k(g)](\psi)\equiv\<\psi|d_k(g)\>:=\overline{\<d_k(g)|\psi\>}$ (see, remark \ref{rem:antidual}).
Note that $\<d_k(g)|\psi\>=\<b_k|(Y\psi)(g)\>_{\mi}$ (see, \eqref{eq:tarkea}) and
$
\<d_k(g)|\psi\>=\<d_k(e)|V_g^*\psi\>.
$
To conclude, we have a minimal diagonalization for $\Mo$:
$$
\<\fii|\Mo(X)\psi\>=\sum_{k=1}^{\rank\Mo}\int_X  \<V^*_g\fii|d_k(e)\>\<d_k(e)|V^*_g\psi\>\d g,\hspace{0.5cm}\fii,\,\psi\in\hB.
$$

For example, when $G\cong\R^n$ (and hence $\hat G\cong\R^n$) one has $\lambda(g)=e^{i\boldsymbol{\lambda}\cdot\mathbf{g}}$ where we have identified $\lambda\in\hat G$ and $g\in G$ with $\boldsymbol{\la}\in\R^n$ and $\mathbf{g}\in\R^n$, respectively. The Haar measures $\d\lambda$ and $\d g$ are now the scaled $n$-dimensional Lebesgue measures $\d\boldsymbol{\la}$ and $(2\pi)^{-n}\d\mathbf{g}$, respectively.
Identifying $G$ with the {\it position space}  and choosing $\hi(\lambda)\equiv\hi_s$
we get the {\it covariant position observables} $\Mo:\,\bo{\R^n}\to L^2_{\hi_s}(\G)$ of a nonrelativistic particle with a spin $s$ (i.e.\ $2s+1=\dim\hi_s$) in the {\it momentum space} $\hat G$. For the canonical position observable $\Qo$, $\mi=\hi_s$, $\rank\Qo=2s+1$, $W=I_{L^2_{\hi_s}(\G)}$, and
\begin{eqnarray*}
\<\fii|\Qo(X)\psi\>&=&
\int_X\int_\G\int_\G e^{i(\boldsymbol{\lambda}-\boldsymbol{\lambda}')\cdot\mathbf{g}}
\<\fii(\boldsymbol{\lambda})|\psi(\boldsymbol{\lambda}')\>_{\mi}\d\boldsymbol{\lambda}\d\boldsymbol{\lambda}' \frac{\d\mathbf{g}}{(2\pi)^n} \\
&=&\sum_{k=1}^{2s+1}
\int_X\int_\G\int_\G e^{i(\boldsymbol{\lambda}-\boldsymbol{\lambda}')\cdot\mathbf{g}}
\<\fii(\boldsymbol{\lambda})|b_k\>_\mi\<b_k|\psi(\boldsymbol{\lambda}')\>_{\mi}\d\boldsymbol{\lambda}\d\boldsymbol{\lambda}' \frac{\d\mathbf{g}}{(2\pi)^n} \\
&=&
\sum_{k=1}^{2s+1}\int_X  \<\fii|d_k(\mathbf{g})\>\<d_k(\mathbf{g})|\psi\>\frac{\d\mathbf{g}}{(2\pi)^n},\hspace{0.5cm}\fii,\,\psi\in\hB,\;X\in\bo G,
\end{eqnarray*}
where the generalized coherent state $d_k(\mathbf{g})\in\hB^\times$ is defined by
$$
\<d_k(\mathbf{g})|\psi\>=
\int_\G e^{-i\boldsymbol{\lambda}\cdot\mathbf{g}}\<b_k|\psi(\boldsymbol{\lambda})\>_{\mi}\d\boldsymbol{\lambda}
=\<b_k|(\hF^*\psi)(\mathbf{g})\>,\hspace{0.5cm}\psi\in\hB.
$$
In the position representation, $\<d_k(\mathbf{g})|$ is the evaluation $\<b_k|\fii(\mathbf{g})\>$ of $\<b_k|\fii\>$ at the point $\mathbf{g}$ which is well-defined, e.g., for the elements $\fii$ of the Schwarz space of rapidly decreasing functions. Then $\fii\mapsto\fii(\mathbf{g})$ is just the (vector valued) Dirac $\delta$-distribution at $\mathbf{g}$.
As a spectral measure, $\Qo$ is extremal.
\end{example}

\section*{Acknowledgements}

This work was supported by the Academy of Finland grant no. 138135.

\end{document}